%% file: main.tex
\documentclass[journal,onecolumn]{IEEEtran}
\usepackage{amsmath,amssymb,amsthm}
\usepackage{enumitem}
\usepackage{standalone} 
\usepackage{bm}
\usepackage{graphics}
\usepackage{mathtools}

\newtheorem{theorem}{Theorem}
\newtheorem{lemma}[theorem]{Lemma}
\newtheorem{proposition}[theorem]{Proposition}
\newtheorem{corollary}[theorem]{Corollary}
\theoremstyle{definition}

\theoremstyle{remark}

\newcommand{\Dxy}{D({Q_{XY}})}
\newcommand{\Dxby}{D({Q_{X\bar{Y}}})}
\newcommand{\Dbxy}{D({Q_{\bar{X}Y}})}
\newcommand{\Dbxby}{D({Q_{\bar{X}\bar{Y}}})}
\newcommand{\Qxy}{Q_{XY}}
\newcommand{\Qxby}{Q_{X\bar{Y}}}
\newcommand{\Qbxy}{Q_{\bar{X}Y}}
\newcommand{\Qbxby}{Q_{\bar{X}\bar{Y}}}
\newcommand{\Pxy}{P_{XY}}
\newcommand{\Pxby}{P_{X\bar{Y}}}
\newcommand{\Pbxy}{P_{\bar{X}Y}}
\newcommand{\Pbxby}{P_{\bar{X}\bar{Y}}}
\newcommand{\xm}{\bm{x}}
\newcommand{\ym}{\bm{y}}
\newcommand{\um}{\bm{u}}
\newcommand{\vm}{\bm{v}}
\newcommand{\Xm}{\bm{X}}
\newcommand{\Ym}{\bm{Y}}
\newcommand{\Um}{\bm{U}}
\newcommand{\Vm}{\bm{V}}
\newcommand{\Zm}{\bm{Z}}
\newcommand{\fx}{f(\xm)}
\newcommand{\fy}{f(\ym)}

\newcommand{\tQxy}{Q_{\xm\ym}}
\newcommand{\tQxby}{Q_{\xm\bar{\ym}}}
\newcommand{\tQbxy}{Q_{\bar{\xm}\ym}}
\newcommand{\tQbxby}{Q_{\bar{\xm}\bar{\ym}}}
\newcommand{\tQxyd}{Q_{\xm'\ym'}}

\newcommand{\tDxy}{D(\tQxy)}
\newcommand{\tDxby}{D(\tQxby)}
\newcommand{\tDbxy}{D(\tQbxy)}
\newcommand{\tDbxby}{D(\tQbxby)}
\newcommand{\Dxyas}{D(\Qxy^*)}
\newcommand{\Dxbyas}{D(\Qxby^*)}
\newcommand{\Dbxyas}{D(\Qbxy^*)}
\newcommand{\Dbxbyas}{D(\Qbxby^*)}

\newcommand{\Dxbymin}{D^{X\bar{Y}}_\mathrm{min}}
\newcommand{\Dbxymin}{D^{\bar{X}Y}_\mathrm{min}}
\newcommand{\Dbxbymin}{D^{\bar{X}\bar{Y}}_\mathrm{min}}

\title{The Condition for Structured Coding to Improve Random Coding in the Binary Modulo-sum Problem}
\author{Yohsuke Tsujino and Shun Watanabe}
\date{\today}

\begin{document}
\maketitle
\begin{abstract}
The modulo-sum problem, proposed by Körner and Marton (KM), is a representative problem in the field of distributed source coding.
In the modulo-sum problem, two correlated sources are encoded separately, and the decoder decodes the modulo-sum of the sources.
It is clear that the Slepian-Wolf (SW) coding rate region is achievable for the modulo-sum problem.
Körner and Marton proved that the SW coding rate region can be improved by structured coding based on linear codes.
Ahlswede and Han (AH) proposed AH coding, which combines structured coding and random coding, and expressed its rate using auxiliary random variables.
However, it was conjectured that the minimum sum rate of AH coding cannot be smaller than the minimum of the sum rates achievable by KM coding or SW coding.
Subsequently, Kakishima and Watanabe considered a multi-letter extension of AH coding, and designed the auxiliary random variables by taking the XOR of adjacent bits of the source sequences.
Through numerical computation, they demonstrated the existence of source parameters for which multi-letter AH coding improves upon SW coding.
However, this confirmation remained numerical, and the conditions under which multi-letter extended AH coding improves upon SW coding have not been analytically characterized.

In this study, we analytically characterize the conditions under which multi-letter extended AH coding improves upon SW coding.
Our condition is tight in the sense that it coincides with the complement of the known SW optimal sufficient condition.
To obtain the conditions, we apply the method of types to reduce the evaluation of the multi-letter expression to a comparison of single-letter divergences, which might be of independent interest.

\end{abstract}
\input{sections/introduction.tex}

\input{sections/preparations.tex}
\input{sections/main_result.tex}
\input{sections/conclusion.tex}
\input{sections/appendix.tex}
\input{sections/acknowledgment.tex}

\bibliographystyle{abbrv}
\bibliography{absref.bib}
    
\end{document}

%% file: sections/introduction.tex
\section{Introduction}
The study of source coding with correlated sources was initiated by 
Slepian and Wolf \cite{slepian:73}. Even though two sources are encoded separately, their result means
the sum of the coding rates can be reduced to the joint entropy of the two sources, which is the fundamental 
limit of the source coding when both the sources are encoded jointly. 
As an extension to lossy source coding, Wyner and Ziv solved the case where one source is available as 
side-information at the decoder \cite{wyner:76}. 
Yamamoto subsequently extended this result to a setting in which the decoder reconstructs a function of the source observed at the encoder and the side information \cite{yamamoto:82}.
For the case in which both sources are encoded, the setting where only one source is reconstructed subject to a distortion constraint was studied by Berger, Housewright, Omura, Tung, and Wolfowitz \cite{BergerHousewright:79}.
For the setting where both sources are reconstructed subject to distortion constraints, the region now known as the Berger–Tung region was derived by Berger \cite{berger:78}, Tung \cite{tung:78}, and Housewright \cite{Housewright:77}.
Furthermore, Yamamoto and Itoh extended the Berger–Tung type region to the problem in which the decoder reconstructs a remote source correlated with the sources observed by the encoders \cite{YamamotoItoh:1980}.
On the other hand, K\"orner and Marton demonstrated that Berger-Tung type coding is not optimal for the 
modulo-sum problem \cite{KM}, in which the decoder is required to reproduce only the modulo-sum of two binary sources. 
This result suggests that the achievable region cannot be fully characterized by standard random coding alone,
and that structured coding must be considered as well.

On the other hand, there are also cases such that Berger-Tung type coding is known
to be optimal. For instance, Gel'fand and Pinsker showed that, when 
the sources satisfy a certain conditional independence condition, Berger-Tung type coding yields the 
optimal achievable rate region \cite{GelPin:79}. This line of research has led to many subsequent developments, including
results on the CEO problem \cite{berger:96,viswanathan:97,oohama:98,oohama:05,PraTseRam:04,wagner:08,wagner:08b,Oohama:06b,oohama:09b,TavVisWag:10,SefTch:16}. 
Regarding the aforementioned modulo-sum problem,
it is known that the Slepian-Wolf region is optimal when the entropy of the function value is at least as large
as the entropy of the sources \cite[Ex 16.23]{Csiszár_Körner_2011}. This condition is equivalent to a degradation-type order between the source
and the function value. For modulo-sum over alphabets of size larger than two, Nair and Wang also showed that
the Slepian-Wolf region is optimal under a broader condition between the source and the function value \cite{Nair}.

There have been attempts to apply the K\"orner-Marton type structured coding to various other network
information problems \cite{NazGas:07,PhiZamEreKhi:11,PraPadShi:monograph}. Nevertheless, compared with standard random coding, our understanding
of structured coding remains immature. In particular, it is difficult to clearly distinguish the situation
in which random coding is sufficient from those in which structured coding is necessary.

In the problem of distributed function computation, two correlated sources are
encoded separately, and the decoder is required to reconstruct a function of the two sources.
Since the decoder can compute the function value if it can reproduce the entire sources,
the Slepian-Wolf achievable region provides a trivial inner bound for function computation.
Han and Kobayashi derived the necessary and sufficient condition on functions such that
this trivial inner bound is optimal for every source satisfying positivity condition \cite{HK}. 
This line of work was further extended to non i.i.d. sources in \cite{KuzWat15,KuzWat16,Wat20}.
The modulo-sum function does not satisfy the Han-Kobayashi condition, and is therefore a typical example
of a function for which the Slepian-Wolf region can be improved. 
 However, although the result by Han and Kobayashi guarantees that, for any function not satisfying their condition,
 there exists at least one source for which the Slepian-Wolf region can be improved, it does not determine whether 
 such an improvement is possible for a given source. In particular, for the modulo-sum function,
 characterizing whether the Slepian-Wolf region can be improved for a given source remains an open problem. 
 
 Recently, goal-oriented communication has attracted considerable attention as a next generation
 communication architecture; e.g.,~see \cite{Gun+23}. Unlike conventional communication schemes, which
 primarily focus on delivering all the data to the receiver, goal-oriented communication
 aims to improve communication efficiency by transmitting only the minimum amount of information
 necessary for the receiver to accomplish a desired task. The distributed function computation problem
 discussed above is a typical example of goal-oriented communication. However, if the Slepian-Wolf region
 cannot be improved, then that means that no gain from goal-oriented communication can be obtained. 
 Therefore, determining whether the Slepian-Wolf region can be improved for a given source
 is an important problem for asserting the usefulness of goal-oriented communication.
 
 When the sources follow a symmetric distribution known as the binary doubly symmetric source (BDSS),
 the K\"orner-Marton coding scheme improves upon the Slepian-Wolf region; in fact, the K\"orner-Marton coding is optimal
 in this case. However, when the distribution deviates from the BDSS, the achievable rate region of the K\"orner-Marton coding
 scheme can even be worse than that of Slepian-Wolf coding. Ahlswede and Han proposed a coding scheme that combines 
 Berger-Tung type random coding and K\"orner-Marton type structured coding in a hybrid manner, and derived an achievable region \cite{AH}.
 Although the Ahlswede-Han region can improve upon the convex hull of the Slepian-Wolf region and the K\"orner-Marton region,
 it has been conjectured by Sefidgaran, Gohari, and Aref \cite{conjecture}
 that the Ahlswede-Han sum rate cannot improve upon the minimum of the Slepian-Wolf sum rate and the K\"orner-Marton sum rate.

In the search for coding schemes that can outperform the sum rate of the Slepian-Wolf region
for distributions such that K\"orner-Marton coding is ineffective, Kakishima and Watanabe numerically demonstrated
that multi-letter Ahlswede-Han coding can improve upon the sum rate of the Slepian-Wolf region for certain source
parameters by choosing the auxiliary random variables as XORs of adjacent source bits \cite{KW}.
In fact, outside the parameter region where the optimality of the Slepian-Wolf region mentioned above is known, 
their numerical experiments suggested a tendency that improvements can be obtained by
increasing the multi-letter block length. However, these findings were supported only by numerical evidence 
and were not accompanied by an analytical proof.

In this work, we analytically characterize the conditions under which multi-letter Ahlswede-Han coding improves upon Slepian-Wolf coding.
Our condition is tight in the sense that it coincides with the complement of the source parameter region where the optimality of Slepian-Wolf coding is known, 
as stated in \cite[Ex. 16.23]{Csiszár_Körner_2011}. 
Thus, the known sufficient condition for the optimality of Slepian-Wolf coding is also necessary.

In our analysis, 
we choose the auxiliary random variables as the XORs of adjacent source bits, 
following the construction in \cite{KW}. The numerical results of Kakishima and Watanabe suggest that, 
as the source block length increases, the parameter region where improvement is observed becomes larger. 
As the block length increases, the sum rate of the multi-letter Ahlswede--Han scheme converges to the Slepian--Wolf sum rate. 
Hence, the key question is not the value of the limit itself, but the side from which the sum rate approaches this limit. 
If the convergence occurs from below, then the multi-letter Ahlswede--Han scheme strictly improves upon the Slepian--Wolf sum rate for sufficiently large block length. Motivated by this observation, we characterize the source distributions for which this convergence occurs from below.
However, directly evaluating this asymptotic behavior is not straightforward. 
To derive the condition, we apply the method of types and reduce the evaluation of the multi-letter expression to a comparison of divergence exponents. 
More specifically, the method of types allows us to identify the dominant terms, and the improvement condition is obtained by comparing the relevant divergences. 
This reduction may be of independent interest.

%% file: sections/preparations.tex
\section{Preparations}
In this paper, we follow the notation from \cite{Csiszár_Körner_2011}.
For example, the entropy of $X$ is denoted by $H(X)$, 
the conditional entropy is denoted by $H(X|Y)$, 
the mutual information is denoted by $I(X\wedge Y)$, 
the binary entropy function is denoted by $h(\cdot)$,
and etc.

The modulo-sum problem, which was first proposed in \cite{KM}, is a distributed source coding problem  where two correlated sources $X$ and $Y$ are encoded separately, and the decoder reproduces the modulo-sum $Z$ of the sources.

Let $(X,Y)$ be a pair of correlated binary sources with joint distribution $P_{XY}$:
\begin{align*}
    P_{XY} = \begin{bmatrix}P_{XY}(0,0) & P_{XY}(0,1) \\ P_{XY}(1,0) & P_{XY}(1,1)\end{bmatrix}=\begin{bmatrix} \theta_{00}& \theta_{01} \\ \theta_{10} & \theta_{11}\end{bmatrix}, \quad \sum_{i,j} \theta_{ij}=1, \quad \theta_{ij}\geq 0.
\end{align*}
Let $P_{X\bar Y}$, $P_{\bar X Y}$, and $P_{\bar X\bar Y}$
denote the distributions obtained from $P_{XY}$ by flipping $Y$, $X$, and both, respectively.

Let $(X^n,Y^n)=((X_1,Y_1),\ldots,(X_n,Y_n))$ be an i.i.d. source generated according to $P_{XY}$, and let
\begin{align*}
    Z^n=(Z_1,\ldots,Z_n), \qquad Z_i=X_i\oplus Y_i
\end{align*}
for $i=1,\ldots,n$.

A code for computing the modulo-sum consists of two encoders
\begin{align*}
    f_n:\{0,1\}^n\to \mathcal{M}_{X,n},
    \qquad
    g_n:\{0,1\}^n\to \mathcal{M}_{Y,n},
\end{align*}
and a decoder
\begin{align*}
    \psi_n:\mathcal{M}_{X,n}\times \mathcal{M}_{Y,n}
    \to \{0,1\}^n .
\end{align*}
The probability of error is defined as
\begin{align*}
    P_{e,n}
    =
    \Pr\left\{
        \psi_n(f_n(X^n),g_n(Y^n))\neq Z^n
    \right\}.
\end{align*}

A rate pair $(R_X,R_Y)$ is said to be achievable if there exists
a sequence of codes
\begin{align*}
    \{(f_n,g_n,\psi_n)\}_{n=1}^{\infty}
\end{align*}
such that
\begin{align*}
    \limsup_{n\to\infty}\frac{1}{n}\log|\mathcal{M}_{X,n}|
    \leq R_X,
\end{align*}
\begin{align*}
    \limsup_{n\to\infty}\frac{1}{n}\log|\mathcal{M}_{Y,n}|
    \leq R_Y,
\end{align*}
and
\begin{align*}
    \lim_{n\to\infty}P_{e,n}=0.
\end{align*}
In this paper, logarithms are to the base 2.

Since the decoder can compute $Z$ from $X$ and $Y$, the Slepian-Wolf (SW) rate region is a trivial achievable rate region for the modulo-sum problem, i.e., a rate pair $(R_X,R_Y)$ satisfying
\begin{equation}
    \begin{aligned}\label{rate:SW}
            R_X &\geq H(X|Y),\\
            R_Y &\geq H(Y|X),\\
            R_X + R_Y &\geq H(X,Y)
        \end{aligned}
\end{equation}        
is achievable.
We denote the optimal sum rate of SW coding by $R_{\mathrm{SW}}=H(X,Y)$.

While SW coding does not consider any structure of the function to be computed, there are known cases where SW coding is optimal \cite[Ex 16.23]{Csiszár_Körner_2011}.
\begin{proposition}
    SW rate region (\ref{rate:SW}) characterizes the achievable region of modulo-sum problem if
    \begin{equation}
        H(Z) \geq\min [H(X),H(Y)].
    \end{equation}
\end{proposition} 
    
Rewriting this condition in terms of $\theta_{ij}$, we have
\begin{equation}\label{cond:SW optimal parameter}
    (\theta_{00}-\theta_{10})(\theta_{01}-\theta_{11}) \geq 0 \quad \text{or} \quad (\theta_{00}-\theta_{01})(\theta_{10}-\theta_{11}) \geq 0.
\end{equation}

Since individual rates must satisfy 
\begin{align*}
    R_X &\geq H(X|Y),\\
    R_Y &\geq H(Y|X)
\end{align*}
even if either $X$ or $Y$ is directly observed at the decoder, 
the SW rate region is suboptimal if and only if the sum rate $R_{\mathrm{SW}}=H(X,Y)$ can be improved.

In \cite{KM}, Körner and Marton (KM) proposed a structured coding based on linear codes, which we call KM coding, and showed the following result:
\begin{proposition}
    A rate pair $(R_X,R_Y)$ satisfying 
    \begin{equation}        
        \begin{aligned}\label{rate:KM}
            R_X &\geq H(Z),\\
            R_Y &\geq H(Z),\\
            R_X + R_Y &\geq 2H(Z)
        \end{aligned}
    \end{equation}
    is achievable.
    The rate region is optimal for symmetric sources, i.e., $\theta_{00}=\theta_{11}$ and $\theta_{01}=\theta_{10}$.
\end{proposition}
In \cite{Nair}, Nair and Wang showed a broader condition for the optimality of KM coding.
    
In \cite{AH}, Ahlswede and Han proposed AH coding which combines structured coding and random coding, and expressed its rate pair using auxiliary random variables.
\begin{proposition}
    A rate pair $(R_X,R_Y)$ satisfying
    \begin{equation}
        \begin{aligned}\label{rate:AH}
            R_X &\geq I(U \wedge X|V)+H(Z|U,V),\\
            R_Y &\geq I(V \wedge Y|U)+H(Z|U,V),\\
            R_X + R_Y &\geq I(X,Y\wedge U,V) + 2H(Z|U,V)
        \end{aligned}
    \end{equation}
    is achievable, where $U$ and $V$ are finite valued random variables satisfying the Markov chain $U-X-Y-V$.
\end{proposition}
However, it was conjectured that the minimum sum rate of AH coding cannot be smaller than the minimum of the sum rates achievable by the KM coding or SW coding \cite{conjecture}.

In \cite{KW}, Kakishima and Watanabe considered a $m$-letter extension to AH coding, and designed the auxiliary random variables by taking the XOR of adjacent bits of the source sequences.
By applying AH scheme for $m$-letter extended sources, we can show that a rate pair $(R_X,R_Y)$ satisfying the following conditions is achievable :
\begin{align*}
    R_X &\geq \frac{1}{m}\left(I(\Um \wedge \Xm|\Vm)+H(\Zm|\Um,\Vm)\right),\\
    R_Y &\geq \frac{1}{m}\left(I(\Vm \wedge \Ym|\Um)+H(\Zm|\Um,\Vm)\right),\\
    R_X + R_Y &\geq \frac{1}{m}\left(I(\Xm,\Ym\wedge \Um,\Vm) + 2H(\Zm|\Um,\Vm)\right),
\end{align*}   
where $\Um$ and $\Vm$ satisfy the Markov chain $\Um-\Xm-\Ym-\Vm$, and $Z_i=X_i \oplus Y_i$ for $i=1,2,\ldots,m$.
For notational simplicity, we write $\Xm,\Ym,\Zm,\Um,\Vm,\xm,\ym,\um,\vm$ for $X^m,Y^m,Z^m,U^{m-1},V^{m-1},x^m,y^m,u^{m-1},v^{m-1}$, respectively.
We denote the sum rate of $m$-letter AH coding whose $\Um,\Vm$ are defined as
\begin{equation*}
U_i=X_i\oplus X_{i+1},  V_i=Y_i\oplus Y_{i+1}, \quad i=1,2,\ldots,m-1    
\end{equation*}
by $R^{(m)}_{\mathrm{AH}}$.
For convenience, we define $R_{\mathrm{AH}}^{(1)}=2H(Z)$.
In the rest of the paper, when we refer to the sum rate of $m$-letter AH coding, we mean $R^{(m)}_{\mathrm{AH}}$.
Kakishima and Watanabe demonstrated the existence of the source parameters where $m$-letter AH coding improves upon SW and KM coding through numerical computation.
In addition, as the letter size $m$ increases, the improvement region expands.
But the conditions under which $m$-letter AH coding improves upon SW coding have not been analytically characterized.

%% file: sections/main_result.tex
\section{Main Result}
\subsection{Statement of Results}
    We analytically characterize the condition for $m$-letter AH coding to improve upon SW coding as follows. 
    \begin{theorem}\label{thm:main}
        If the following condition holds
        \begin{equation}\label{cond:mAH improves SW}
            (\theta_{00}-\theta_{10})(\theta_{01}-\theta_{11}) < 0 \quad and \quad (\theta_{00}-\theta_{01})(\theta_{10}-\theta_{11}) < 0,
        \end{equation}
        there exists an integer $m_0$ such that the sum rate $R^{(m)}_{\mathrm{AH}}$ 
        of $m$-letter AH coding is strictly smaller than 
        the SW sum rate $R_{\mathrm{SW}}$ for every $m\geq m_0$.
    \end{theorem}

    It should be noted that $R_{\mathrm{AH}}^{(m)}$ converges to $H(X,Y)$ as $m$ increases.
    Thus, although Theorem \ref{thm:main} guarantees a strict improvement over the SW sum rate
    for all sufficiently large $m$, the amount of improvement decreases as $m$ increases.
    According to the numerical results in \cite{KW}, a strict improvement can be
    observed for relatively small values of $m$ over a fairly broad range of parameters.
    However, as can be seen from the proof of Theorem \ref{thm:main}, 
    the smallest positive integer $m_0$ for which a strict improvement 
    can be theoretically guaranteed is quite large.
    Even if an improvement is observed for a small value of $m$,
    it is unclear whether the improvement persists for every larger $m$.
    However, Theorem \ref{thm:main} guarantees that the improvement continues 
    for all $m$ beyond some threshold $m_0$.

    Perhaps surprisingly, the condition (\ref{cond:mAH improves SW}) is exactly the complement of the condition \eqref{cond:SW optimal parameter} for the optimality of SW coding.
    In other words, the sufficient condition for the SW rate region to be optimal shown in \cite[Ex 16.23]{Csiszár_Körner_2011} is also necessary.
    \begin{corollary}\label{cor:optimality}
        The SW rate region is optimal if and only if \eqref{cond:SW optimal parameter} holds.
    \end{corollary}
    
    \begin{proof}
        The known sufficient condition for the SW rate region to be optimal is given by \eqref{cond:SW optimal parameter}.
        Theorem \ref{thm:main} shows that the $m$-letter AH coding strictly improves the SW sum rate when \eqref{cond:mAH improves SW} holds.
        Since condition \eqref{cond:mAH improves SW} is exactly the complement of condition \eqref{cond:SW optimal parameter}, the sufficient condition for the SW rate region to be optimal is also necessary.
    \end{proof}

    When the source distribution is given by $\theta_{00}=a,\theta_{01}=\theta_{10}=b,\theta_{11}=1-a-2b$,
    we can rewrite the condition \eqref{cond:mAH improves SW} in terms of $a$ and $b$ as
    \begin{equation}
    \begin{aligned}\label{cond:mAH improves SW ab}
        (a>b,1-a-2b>b) \quad \text{or} \quad (a<b,1-a-2b<b).
    \end{aligned}
    \end{equation}
    Figure \ref{fig:example} plots the regions of the parameter $(a,b)$ in which SW coding is optimal and sub-optimal.

    \begin{figure}[t]
        \centering
        \includegraphics[width=0.8\linewidth]{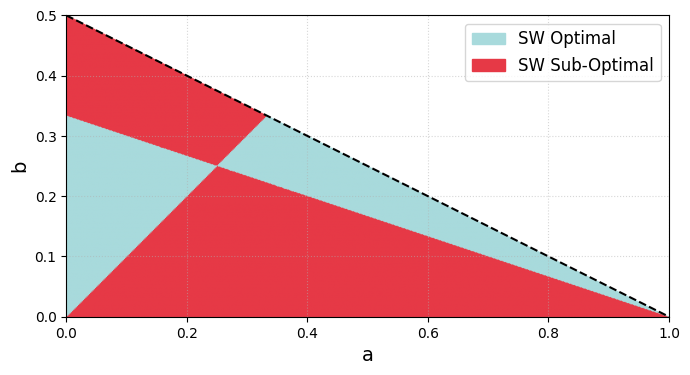}
        \caption{A description of the condition on $a,b$ in \eqref{cond:mAH improves SW ab}}
    \label{fig:example}
    \end{figure}

\subsection{Proof}
    In this subsection, we give a proof of Theorem \ref{thm:main}.

    Recall that we analyze the performance of $m$-letter AH coding with the auxiliary random variables defined as adjacent parity, i.e., $U_i=X_i\oplus X_{i+1}$ and $V_i=Y_i\oplus Y_{i+1}$ for $i=1,2,\ldots,m-1$.
    This construction gives rise to a particular coset structure as follows.
    For $( \um , \vm )\in \{0,1\}^{m-1} \times \{0,1\}^{m-1}$, let
        \begin{align*}
            \mathcal{A}( \um , \vm ) = \{( \xm , \ym ): x_i \oplus x_{i+1} = u_i, y_i \oplus y_{i+1} = v_i, 1 \leq i \leq m-1 \}.
        \end{align*}
    Similarly, let
        \begin{align*}
            \mathcal{B}( \um , \vm ,z_m) = \{( \xm , \ym ): x_i \oplus x_{i+1} = u_i, y_i \oplus y_{i+1} = v_i, z_m = x_m \oplus y_m, 1 \leq i \leq m-1 \}.
        \end{align*}
    For a binary sequence $\xm$, we denote by $\bar{\xm}$ its bitwise complement, i.e., $\bar{x}_i=x_i\oplus 1$ for all $i$.
    The same notation is used for $\ym$.
    Then we can see that $|\mathcal{A}(\um,\vm)|=4$ for any $( \um , \vm )$, since only $( \xm , \ym ),( \xm ,\bar{ \ym }),(\bar{ \xm }, \ym ),$ and $(\bar{ \xm },\bar{ \ym })$ produce the same $( \um , \vm )$.
    Also, $|\mathcal{B}(\um,\vm,z_m)|=2$ for any $( \um , \vm ,z_m)$, since only $( \xm , \ym )$ and $(\bar{ \xm },\bar{ \ym })$ produce the same $( \um , \vm ,z_m)$.

    To improve upon SW coding, the sum rate of $m$-letter AH coding must be smaller than that of SW coding.
    Thus, we are able to rewrite the improvement condition as 
    \begin{equation*}
            R^{(m)}_{\mathrm{AH}} < R_{\mathrm{SW}}.
    \end{equation*}
    Since $Z_1,\ldots,Z_{m-1}$ are deterministic functions of $(\Um,\Vm,Z_m)$, and since $\Um$ and $\Vm$ are functions of $\Xm$ and $\Ym$, we can rewrite $R^{(m)}_{\mathrm{AH}}$ as
    \begin{align*}
        R^{(m)}_{\mathrm{AH}} = H(X,Y) + \frac{1}{m}\left( H(\Xm,\Ym|\Um,\Vm) - 2H(\Xm,\Ym|\Um,\Vm,Z_m) \right).
    \end{align*}

    On the other hand, $R_{\mathrm{SW}}$ is given by 
    \begin{align*}
        R_{\mathrm{SW}} = H(X,Y).
    \end{align*}
    Therefore, the improvement condition can be rewritten as
    \begin{align}
        C^{(m)} := H(\Xm,\Ym|\Um,\Vm) - 2H(\Xm,\Ym|\Um,\Vm,Z_m) < 0. \label{eq:C}
    \end{align}

    We first dispose of some trivial cases.
    If three of the four source probabilities are zero, then $(X,Y)$ is deterministic.
    Hence $Z$ is also deterministic, and \eqref{cond:SW optimal parameter} is satisfied.

    If exactly two of the four source probabilities are zero, then the support of $P_{XY}$ consists of two points.
    There are two cases.
    If the two points have the same modulo-sum value, then $Z$ is deterministic.
    Thus, zero communication is sufficient to recover $Z$, and the SW sum rate is strictly improved whenever $H(X,Y)>0$.
    This case occurs exactly when the two nonzero probabilities lie on the diagonal or on the off-diagonal, 
    and it satisfies condition \eqref{cond:mAH improves SW}, 
    $R_{\mathrm{AH}}^{(1)}=2H(Z)<H(X,Y)$ and 
    $R_{\mathrm{AH}}^{(m)}=\frac{1}{m}I(\Xm,\Ym\wedge\Um\Vm)<H(X,Y)$ for every $m\geq 2$.

    On the other hand, if the two points have different modulo-sum values, then the two nonzero probabilities lie in the same row or in the same column.
    In this case, one of the sources is deterministic, and recovering $Z$ is equivalent to recovering the other source.
    Hence the SW sum rate is already optimal.
    We can also verify that condition \eqref{cond:SW optimal parameter} holds.

    Therefore, in the rest of the proof, we only need to consider the cases where either $P_{XY}$ has full support or exactly one source probability is zero.

    To analyze the sign of $C^{(m)}$, we first rewrite $C^{(m)}$ in terms of the types of $(\Xm,\Ym)$.
    For notational simplicity, we introduce the following notation for the types and divergences.
    
    Let $\mathcal P_m$ be the set of all joint types on $\{0,1\}^2$ with denominator $m$. For a type $Q_{XY}\in\mathcal P_m$, if $\operatorname{supp}(Q_{XY})\not\subset\operatorname{supp}(P_{XY})$, then the type class $\mathcal{T}_{Q_{XY}}$ has zero probability.
    Equivalently, we use the convention that $D(Q_{XY}|P_{XY})=\infty$ when $\operatorname{supp}(Q_{XY})\not\subset\operatorname{supp}(P_{XY})$, and $2^{-m\infty}=0$.

    For a joint type $\Qxy$, let $\mathcal{T}_{\Qxy}$ denote the set of all pairs of sequences $(\xm,\ym)$ whose joint type is $\Qxy$,
    and we define $\Qxby$, $\Qbxy$ and $\Qbxby$ as the joint types obtained from $\Qxy$ by flipping $Y$, $X$, and both, respectively.
    For a joint distribution $\Qxy$, we define the notation of divergences between $\Qxy$ and $P_{XY}$ as follows:
    \begin{equation}\label{nt:divergence}
        \Dxy \coloneq D(\Qxy||P_{XY}).
    \end{equation}
    Then, we can rewrite $C^{(m)}$ as the following lemma.

    \begin{lemma}\label{lem:type expression of C}
    The quantity $C^{(m)}$ admits the representation
    \begin{equation}\label{eq:type expression of C^{(m)}}
        C^{(m)}=\sum_{Q_{XY} \in \mathcal{P}_m}P(\mathcal{T}_{Q_{XY}})C_{\mathrm{i}}(Q_{XY})
    \end{equation}
    where 
    \begin{align}
        C_\mathrm{i}(Q_{XY}) &= \left[  h\left( \frac{2^{-m\Dxy}+2^{-m\Dbxby}}{S(Q_{XY})} \right)\right.\notag\\
        &\quad -\left(\frac{2^{-m\Dxy}+2^{-m\Dbxby}}{S(Q_{XY})}\right)h\left(\frac{2^{-m\Dxy}}{2^{-m\Dxy}+2^{-m\Dbxby}}\right)\label{eq:C_i} \\ 
        &\quad -\left.\left(\frac{2^{-m\Dxby}+2^{-m\Dbxy}}{S(Q_{XY})}\right)h\left(\frac{2^{-m\Dxby}}{2^{-m\Dxby}+2^{-m\Dbxy}}\right)\right],\notag\\
        S(Q_{XY}) &= 2^{-m \Dxy } + 2^{-m \Dxby } + 2^{-m \Dbxy } + 2^{-m \Dbxby }.\label{eq:S}
    \end{align}

    \end{lemma}
    
    The proof of Lemma \ref{lem:type expression of C} is given in Appendix \ref{prf:type expression of C}.
    Because we get $P(\mathcal{T}_{Q_{XY}})=0$ when $S(Q_{XY})=0$, and thus the contribution of $Q_{XY}$ to $C^{(m)}$ is zero, we can ignore the case where $S(Q_{XY})=0$.

    In fact, $C_{\mathrm{i}}(Q_{XY})$ has the same value for the types $Q_{XY}$, $Q_{X\bar{Y}}$, $Q_{\bar{X}Y}$ and $Q_{\bar{X}\bar{Y}}$.

    Moreover, $C_{\mathrm{i}}(Q_{XY})$ can be rewritten as exponential terms of the divergences $\Dxy$, $\Dxby$, $\Dbxy$ and $\Dbxby$ defined in (\ref{nt:divergence}).
    Considering this symmetry and deriving an upper bound on $C_{\mathrm{i}}(Q_{XY})$, we can get the following lemma which gives us an upper bound of $C^{(m)}$.
    As we will see later, this upper bound is tight in the sense that the sign on $C^{(m)}$ is determined by the dominant exponential term when $m$ is large.

    \begin{lemma}\label{lem:upper bound of C}
        $C^{(m)}$ can be upper bounded as
        \begin{equation}\label{eq:upper bound of C}
            C^{(m)} \leq \sum_{Q_{XY} \in \mathcal{Q}_m} \left\{8\left(2 + m\log \frac{1}{\theta^\star}\right)\left(2^{-m\Dxby}+2^{-m\Dbxy}\right)  - \frac{2^{-m\Dbxby}}{16(m+1)^4}\right\}.
        \end{equation}
        where $\mathcal{Q}_m$ is subset of the set of types defined as
        \begin{equation} \label{eq:mathcalQ_m}
        \mathcal{Q}_m = \{Q_{XY}\in \mathcal{P}_m: \Dxy \leq \min[\Dxby,\Dbxy,\Dbxby]\}
        \end{equation}
        and 
        \begin{equation}
            \theta^\star = \min \{\theta_{ij}:\theta_{ij} > 0\}. \label{nt:theta star}
        \end{equation}
    \end{lemma}
    Proof of Lemma \ref{lem:upper bound of C} is given in Appendix \ref{prf:upper bound of C}.

    The upper bound in Lemma \ref{lem:upper bound of C} becomes negative 
    if a negative term becomes dominant for some joint type $Q_{XY} \in \mathcal{Q}_m$.
    Thus, the improvement condition can be characterized by comparing the minimum of
    $\Dxby$, $\Dbxy$ and $\Dbxby$ among $Q_{XY}\in \mathcal{Q}_m$.

    We now minimize $\Dxby,\Dbxy,\Dbxby$ over $Q_{XY} \in \mathcal{Q}_m$ using the method of Lagrange multipliers.
    To simplify analysis, we first relax the optimization problem by minimizing $\Dxby,\Dbxy,\Dbxby$ over  $\mathcal{Q}$, 
    which is not a type but a distribution, where 
    \begin{equation}
        \mathcal{Q}=\{Q_{XY} \in \mathcal{P} (\operatorname{supp}(P_{XY})): \Dxy \leq \min[\Dxby,\Dbxy,\Dbxby]\}.        
    \end{equation}
    where $\mathcal{P} (\operatorname{supp}(P_{XY}))$ is the set of all distributions on the support of $P_{XY}$.
    In other words, we ignore the condition that $\Qxy$ must be a type.

    \begin{lemma}\label{lem:minimum divergence}
    When $|\operatorname{supp}(P_{XY})|\geq 3$, the minimum values of $\Dxby,\Dbxy,\Dbxby$ over $\mathcal{Q}$ are given by
    \begin{align}
        \Dxbymin \coloneqq \min_{\Qxy \in \mathcal{Q}} \Dxby =\log\frac{1}{2(\sqrt{\theta_{00}\theta_{01}}+\sqrt{\theta_{10}\theta_{11}})},\label{eq:min Dxby} \\ 
        \Dbxymin\coloneqq \min_{\Qxy \in \mathcal{Q}} \Dbxy=\log\frac{1}{2(\sqrt{\theta_{00}\theta_{10}}+\sqrt{\theta_{01}\theta_{11}})},\label{eq:min Dbxy} \\
        \Dbxbymin \coloneqq \min_{\Qxy \in \mathcal{Q}} \Dbxby=\log\frac{1}{2(\sqrt{\theta_{00}\theta_{11}}+\sqrt{\theta_{01}\theta_{10}})}.\label{eq:min Dbxby}
    \end{align}
    Furthermore, the minimum value in \eqref{eq:min Dbxby} is attained at $\Qxy^*$ satisfying
    \begin{align}
        \Dbxbyas = \Dxyas,\label{eq:boundary equality} \\
        \Dxbyas = \Dbxyas.\label{eq:boundary equality2}
    \end{align}
    \end{lemma}
    The way to solve the optimization problem is given in Appendix \ref{prf:optimization}.

    The condition 
    \begin{equation}\label{cond:negative cond in terms of divergence}
        \Dbxbymin < \min\{\Dxbymin, \Dbxymin\}
    \end{equation}
    determines the improvement of $m$-letter AH coding over SW coding.
    In fact, we have
    \begin{align*}
            \Dbxbymin < \Dxbymin
            &\Leftrightarrow
            \sqrt{\theta_{00}\theta_{11}} + \sqrt{\theta_{01}\theta_{10}} > \sqrt{\theta_{00}\theta_{01}} + \sqrt{\theta_{10}\theta_{11}}\\
            &\Leftrightarrow
            (\sqrt{\theta_{00}}-\sqrt{\theta_{10}})(\sqrt{\theta_{01}}-\sqrt{\theta_{11}}) < 0\\
            &\Leftrightarrow
            (\theta_{00}-\theta_{10})(\theta_{01}-\theta_{11}) < 0.
    \end{align*}
    The same argument applies to the other inequality $\Dbxbymin < \Dbxymin$.
    Therefore, the condition \eqref{cond:negative cond in terms of divergence} coincides with
    the condition \eqref{cond:mAH improves SW}.

    Now we check that $C^{(m)}$ is negative for sufficiently large $m$ under condition \eqref{cond:negative cond in terms of divergence}.
    For every $\Qxy\in\mathcal Q_m$, we have
    \begin{align*}
        \Dxby
        &\geq
        \Dxbymin,\\
        \Dbxy
        &\geq
        \Dbxymin.
    \end{align*}
    Note that left-hand side is the divergence between a type and $P_{XY}$, 
    while the right-hand side is the minimum divergence over $\mathcal{Q}$ 
    which is subset of general distribution.\footnote{
        Even though $Q_m$ contains joint types $Q_{XY}$ such that
    $\operatorname{supp}(Q_{XY}) \not\subseteq \operatorname{supp}(P_{XY})$, 
    $D(\Qxby)=\infty$ and $D(\Qbxy)=\infty$ for such $Q_{XY}$.
    }
    The positive terms of \eqref{eq:upper bound of C} in Lemma \ref{lem:upper bound of C} are bounded as
    \begin{align*}
        \sum_{\Qxy\in\mathcal Q_m}
        8\left(2+m\log\frac{1}{\theta^\star}\right)
        \left(
            2^{-m\Dxby}
            +
            2^{-m\Dbxy}
        \right)
        &\leq
        16|\mathcal Q_m|
        \left(2+m\log\frac{1}{\theta^\star}\right)
        2^{-m\min\{\Dxbymin,\Dbxymin\}}.
    \end{align*}
    On the other hand, by the density of types argument in Appendix \ref{prf:density of types},
    there exist a type $Q_{m,XY}\in\mathcal Q_m$ and $\varepsilon_m$ such that $\varepsilon_m \to 0$ as $m \to \infty$ and
    \begin{align*}
        D(Q_{m,\bar X\bar Y})
        \leq
        \Dbxbymin+\varepsilon_m.
    \end{align*}
    Therefore, keeping only the negative term corresponding to $Q_{m,XY}$, we obtain
    \begin{align*}
        C^{(m)}
        &\leq
        16|\mathcal Q_m|
        \left(2+m\log\frac{1}{\theta^\star}\right)
        2^{-m\min\{\Dxbymin,\Dbxymin\}}
        -
        \frac{1}{16(m+1)^4}
        2^{-m(\Dbxbymin+\varepsilon_m)}\\
        &\leq
        16(m+1)^4
        \left(2+m\log\frac{1}{\theta^\star}\right)
        2^{-m\min\{\Dxbymin,\Dbxymin\}}
        -
        \frac{1}{16(m+1)^4}
        2^{-m(\Dbxbymin+\varepsilon_m)}\\
        &=
        \frac{2^{-m(\Dbxbymin+\varepsilon_m)}}{16(m+1)^4}
        \left[
            256(m+1)^8
            \left(2+m\log\frac{1}{\theta^\star}\right)
            2^{-m(\min\{\Dxbymin,\Dbxymin\}-\Dbxbymin-\varepsilon_m)}
            -
            1
        \right].
    \end{align*}
    Since
    \begin{align*}
        \min\{\Dxbymin,\Dbxymin\}-\Dbxbymin>0
    \end{align*}
    and $\varepsilon_m\to0$ as $m\to\infty$, the term in brackets is negative for every sufficiently large $m$.
    Therefore, there exists $m_0$ such that
    \begin{align*}
        C^{(m)}<0
    \end{align*}
    for every $m\geq m_0$.

%% file: sections/conclusion.tex
\section{Conclusion}
  In this paper, we analytically characterized the condition under which the $m$-letter AH scheme improves upon the SW sum rate for the modulo-sum problem. 
  In previous work, the improvement region had only been observed through numerical computations. 
  In contrast, we derived an explicit condition by applying the method of types to the rate difference between $m$-letter AH and SW coding.
  
  More specifically, we rewrote the rate difference as a sum over joint types and reduced the sign analysis to a comparison of divergence exponents. 
  Since the divergence is a convex and tractable quantity, this reduction enabled us to determine the improvement region explicitly. 
  As a consequence, we showed that the known sufficient condition for the optimality of the SW sum rate in the modulo-sum problem is also necessary.

  There remain several directions for future work. First, the achievable sum rate depends on the choice of the auxiliary random variables. 
  In this paper, we considered the auxiliary random variables induced by the parity of adjacent source symbols. 
  However, it is not clear whether this construction is optimal. 
  Finding better auxiliary random variables may lead to further improvements in the achievable sum rate and determining the optimal rate region of the modulo-sum problem.

%% file: sections/appendix.tex
\appendix
    \subsection{Proof of Lemma \ref{lem:type expression of C}} \label{prf:type expression of C}
        
    We analyze the two entropy terms appearing in $C^{(m)}$.
    For simplicity, we write conditional entropy $H(\Xm,\Ym|\Um=\um,\Vm=\vm)$ as $H(\Xm,\Ym|\um,\vm)$, and $H(\Xm,\Ym|\Um=\um,\Vm=\vm,Z_m=z_m)$ as $H(\Xm,\Ym|\um,\vm,z_m)$.
    
    For the first term of \eqref{eq:C}, we have
    \begin{align*}
        &H(\Xm,\Ym|\Um,\Vm)\\
        &= \sum_{ \um , \vm } P_{\Um\Vm}( \um , \vm ) H(\Xm,\Ym| \um , \vm )\\
        &= \sum_{ \um , \vm } \sum_{( \xm , \ym )\in \mathcal{A}( \um , \vm )} P_{\Xm\Ym}( \xm , \ym ) H(\Xm,\Ym| \um , \vm ) \\
        &= \sum_{ \um , \vm } \sum_{( \xm , \ym )\in \mathcal{A}( \um , \vm )} P_{\Xm\Ym}( \xm , \ym ) H(\Xm,\Ym| f ( \xm ), f ( \ym ))\\
        &= \sum_{ \xm , \ym }P_{\Xm\Ym}( \xm , \ym ) H(\Xm,\Ym| f ( \xm ), f ( \ym ))\\
        &= \sum_{Q_{XY}\in \mathcal{P}_m} \sum_{( \xm , \ym )\in \mathcal{T}_{Q_{XY}} } P_{\Xm\Ym}( \xm , \ym )H(\Xm,\Ym| f ( \xm ), f ( \ym ))
    \end{align*}
    where $f$ is the function from a source sequence to the corresponding auxiliary sequence, defined as
    \begin{equation*}
        f(\xm)
        =
        (x_1\oplus x_2,\ldots,x_{m-1}\oplus x_m),
    \end{equation*}
    and \(f(\ym)\) is defined in the same way.
    Similarly, for the second term of \eqref{eq:C}, we have
    \begin{align*}
        &H(\Xm,\Ym|\Um,\Vm,Z_m)\\
        &= \sum_{ \um , \vm ,z_m} P_{\Um\Vm Z_m}( \um , \vm ,z_m) H(\Xm,\Ym| \um , \vm ,z_m)\\
        &= \sum_{ \um , \vm } P_{\Um\Vm}( \um , \vm ) \sum_{z_m} P_{Z_m|\Um\Vm}(z_m| \um , \vm ) H(\Xm,\Ym| \um , \vm ,z_m)\\
        &= \sum_{ \um , \vm } P_{\Um\Vm}( \um , \vm ) H(\Xm,\Ym| \um , \vm ,Z_m),\\
        \intertext{    where $H(\Xm,\Ym| \um , \vm ,Z_m)=\sum_{z_m} P_{\Zm|\Um\Vm}(z_m| \um , \vm ) H(\Xm,\Ym| \um , \vm ,z_m)$.}
        \intertext{Considering the same manner as the first term of \eqref{eq:C}, we get}
        &\sum_{ \um , \vm } P_{\Um\Vm}( \um , \vm ) H(\Xm,\Ym| \um , \vm ,Z_m)\\
        &= \sum_{Q_{XY}\in \mathcal{P}_m} \sum_{( \xm , \ym )\in \mathcal{T}_{Q_{XY}} } P_{\Xm\Ym}( \xm , \ym )H(\Xm,\Ym| f ( \xm ), f ( \ym ),Z_m).
    \end{align*}
    To evaluate the above entropies, we need the following probabilities:
    \begin{equation}
        \begin{aligned}\label{eq:probabilities}
        P_{\Xm\Ym|\Um\Vm}( \xm , \ym | \fx , \fy ) &= \frac{P_{\Xm\Ym}( \xm , \ym )}{\sum_{( \xm' , \ym' )\in \mathcal{A}( \fx , \fy )} P_{\Xm\Ym}( \xm' , \ym' )}\\
        &\overset{(a)}{=} \frac{2^{-m(H(\tQxy)+D(\tQxy||P_{XY}))}}{\sum_{( \xm' , \ym' )\in \mathcal{A}( \fx , \fy )} 2^{-m(H(\tQxyd)+D(\tQxyd||P_{XY}))}}\\
        &\overset{(b)}{=} \frac{2^{-mD(\tQxy||P_{XY})}}{\sum_{( \xm' , \ym' )\in \mathcal{A}( \fx , \fy )} 2^{-mD(\tQxyd||P_{XY})}}\\
        &=\frac{2^{-mD(\tQxy||P_{XY})}}{S(\tQxy)},\\
        P_{Z_m|\Um\Vm}( z_m | \fx , \fy ) &= \frac{\sum_{( \xm' , \ym' )\in \mathcal{B}( \fx , \fy ,z_m)} P_{\Xm\Ym}( \xm' , \ym' )}{\sum_{( \xm' , \ym' )\in \mathcal{A}( \fx , \fy )} P_{\Xm\Ym}( \xm' , \ym' )}\\
        &\overset{(c)}{=} \frac{\sum_{( \xm' , \ym' )\in \mathcal{B}( \fx , \fy ,z_m)} 2^{-m(H(\tQxyd)+D(\tQxyd||P_{XY}))}}{\sum_{( \xm' , \ym' )\in \mathcal{A}( \fx , \fy )} 2^{-m(H(\tQxyd)+D(\tQxyd||P_{XY}))}}\\
        &\overset{(d)}{=} \frac{\sum_{( \xm' , \ym' )\in \mathcal{B}( \fx , \fy ,z_m)} 2^{-mD(\tQxyd||P_{XY})}}{\sum_{( \xm' , \ym' )\in \mathcal{A}( \fx , \fy )} 2^{-mD(\tQxyd||P_{XY})}}\\
        &= 
        \begin{cases}
            \frac{2^{-mD(\tQxy||P_{XY})}+2^{-mD(\tQbxby||P_{XY})}}{S(\tQxy)} &\text{if } z_m=g(\xm,\ym),\\
            \frac{2^{-mD(\tQxby||P_{XY})}+2^{-mD(\tQbxy||P_{XY})}}{S(\tQxy)} &\text{if }z_m\neq g(\xm,\ym),
        \end{cases}\\
        P_{\Xm\Ym|\Um\Vm Z_m}( \xm , \ym | \fx , \fy , g(\xm,\ym)) &= \frac{P_{\Xm\Ym}( \xm , \ym )}{\sum_{( \xm' , \ym' )\in \mathcal{B}( \fx , \fy ,z_m)} P_{\Xm\Ym}( \xm' , \ym' )}\\
        &\overset{(e)}{=} \frac{2^{-m(H(\tQxy)+D(\tQxy||P_{XY}))}}{\sum_{( \xm' , \ym' )\in \mathcal{B}( \fx , \fy ,z_m)} 2^{-m(H(\tQxyd)+D(\tQxyd||P_{XY}))}}\\
        &\overset{(f)}{=} \frac{2^{-mD(\tQxy||P_{XY})}}{\sum_{( \xm' , \ym' )\in \mathcal{B}( \fx , \fy ,z_m)} 2^{-mD(\tQxyd||P_{XY})}}\\
        &= \frac{2^{-mD(\tQxy||P_{XY})}}{2^{-mD(\tQxy||P_{XY})}+2^{-mD(\tQbxby||P_{XY})}}.\\
        \end{aligned}
    \end{equation}
    where $S(\tQxy)$ is defined in \eqref{eq:S},
    $g$ is the function defined as $g(\xm,\ym)=x_m\oplus y_m$,
    (a), (c), and (e) follow from a standard result in the method of types \cite[Lemma 2.6]{Csiszár_Körner_2011}, which states that the probability of a sequence $(\xm, \ym)$ with joint type $\tQxy$ under $P_{XY}$ is given by
    \begin{equation}\label{eq:probability of type}
        P_{\Xm\Ym}( \xm , \ym ) = 2^{-m(H(\tQxy)+D(\tQxy||P_{XY}))},
    \end{equation}
    and (b), (d), and (f) follow from the fact that $H(\tQxy),H(\tQxby),H(\tQbxy)$ and $H(\tQbxby)$ are equal since the four types are obtained from each other by permutation of the alphabet symbols.

    Substituting probability expressions \eqref{eq:probabilities} into the entropies, we get
    \begin{align*}
        H(\Xm,\Ym| \fx,\fy )=H\left[\frac{2^{-m \tDxy }}{S(\tQxy)},\frac{2^{-m \tDxby }}{S(\tQxy)},\frac{2^{-m \tDbxy }}{S(\tQxy)},\frac{2^{-m \tDbxby }}{S(\tQxy)}\right].
    \end{align*}
    Note that $\tDxy,\tDxby,\tDbxy,\tDbxby$ are functions of $\tQxy$.
    In terms of the binary entropy function, this can be expressed as,
    \begin{align*}
        &H(\Xm,\Ym| \fx,\fy )\\
        &= h\left(\frac{2^{-m \tDxy }+ 2^{-m \tDbxby }}{S(\tQxy)}\right) + \left(\frac{2^{-m \tDxy }+ 2^{-m \tDbxby }}{S(\tQxy)}\right)h\left(\frac{2^{-m \tDxy }}{2^{-m \tDxy }+ 2^{-m \tDbxby }}\right) \\
        &\quad + \left(\frac{2^{-m \tDxby }+ 2^{-m \tDbxy }}{S(\tQxy)}\right)h\left(\frac{2^{-m \tDxby }}{2^{-m \tDxby }+ 2^{-m \tDbxy }}\right).
    \end{align*}
    Similarly,
    \begin{align*}
        &H(\Xm,\Ym| \fx,\fy ,Z_m)\\
        &= \sum_{z_m} P_{Z_m|\Um\Vm}(z_m| \fx,\fy ) H(\Xm,\Ym| \fx,\fy ,z_m)\\
        &= \sum_{z_m} P_{Z_m|\Um\Vm}(z_m| \fx,\fy )\sum_{ (\xm' , \ym')  \in \mathcal{B}(\fx,\fy,z_m)} P_{\Xm\Ym|\Um\Vm Z_m}( \xm' , \ym' | \fx,\fy ,z_m) \log \frac{1}{P_{\Xm\Ym|\Um\Vm Z_m}( \xm' , \ym' | \fx,\fy ,z_m)}\\
        &= \sum_{z_m} P_{Z_m|\Um\Vm}(z_m| \fx,\fy )\sum_{ (\xm' , \ym')  \in \mathcal{B}(\fx,\fy,z_m)} P_{\Xm\Ym|\Um\Vm Z_m}( \xm' , \ym' | f(\xm'),f(\ym') ,z_m) \log \frac{1}{P_{\Xm\Ym|\Um\Vm Z_m}( \xm' , \ym' | f(\xm'),f(\ym'),z_m)}\\
        &= \left( \frac{2^{-m\tDxy}+2^{-m\tDbxby}}{S(\tQxy)} \right)h\left( \frac{2^{-m \tDxy }}{2^{-m \tDxy }+ 2^{-m \tDbxby }}\right) + \left( \frac{2^{-m\tDxby}+2^{-m\tDbxy}}{S(\tQxy)} \right)h\left( \frac{2^{-m \tDxby }}{2^{-m \tDxby }+ 2^{-m \tDbxy }}\right).
    \end{align*}
    Thus, we have
    \begin{align*}
        C^{(m)} &= \sum_{\Qxy}\sum_{( \xm , \ym )\in \mathcal{T}_{\Qxy}} P_{\Xm\Ym}( \xm , \ym )\left\{  h\left( \frac{2^{-m\tDxy}+2^{-m\tDbxby}}{S(\tQxy)} \right)\right.\\
        &\quad-\left(\frac{2^{-m\tDxy}+2^{-m\tDbxby}}{S(\tQxy)}\right)h\left(\frac{2^{-m\tDxy}}{2^{-m\tDxy}+2^{-m\tDbxby}}\right)\\
        &\quad\left.-\left(\frac{2^{-m\tDxby}+2^{-m\tDbxy}}{S(\tQxy)}\right)h\left(\frac{2^{-m\tDxby}}{2^{-m\tDxby}+2^{-m\tDbxy}}\right)\right\}\\
        &= \sum_{\Qxy} P(\mathcal{T}_{\Qxy}) \left\{ h\left( \frac{2^{-m\Dxy}+2^{-m\Dbxby}}{S(\Qxy)} \right)\right.\\
        &\quad-\left(\frac{2^{-m\Dxy}+2^{-m\Dbxby}}{S(\Qxy)}\right)h\left(\frac{2^{-m\Dxy}}{2^{-m\Dxy}+2^{-m\Dbxby}}\right)\\
        &\quad\left.-\left(\frac{2^{-m\Dxby}+2^{-m\Dbxy}}{S(\Qxy)}\right)h\left(\frac{2^{-m\Dxby}}{2^{-m\Dxby}+2^{-m\Dbxy}}\right)\right\}.
    \end{align*}
    Substituting $C_{\mathrm{i}}(\Qxy)$ defined in \eqref{eq:C_i}, finally, we get type expression of $C^{(m)}$ as \eqref{eq:type expression of C^{(m)}}.
    \qed

    \subsection{Proof of Lemma \ref{lem:upper bound of C}}\label{prf:upper bound of C}
    Although $P_{XY}$ may not have full support, we take
    $\mathcal{P}_m$ to be the set of all joint types on
    $\mathcal{X}\times\mathcal{Y}$.  If a type is not absolutely
    continuous with respect to $P_{XY}$, its divergence from
    $P_{XY}$ is defined to be infinity, and we use the convention
    $2^{-m\infty}=0$.  With this convention, the set
    $\mathcal{P}_m$ is closed under flipping $X$, $Y$, or both.
    Therefore, the four types
    \[
        Q_{XY},\quad Q_{X\bar Y},\quad Q_{\bar X Y},\quad
        Q_{\bar X\bar Y}
    \]
    can be treated as one equivalence class even when $P_{XY}$ does
    not have full support.

    It is clear that $C_{\mathrm{i}}(\Qxy),C_{\mathrm{i}}(\Qxby),C_{\mathrm{i}}(\Qbxy),C_{\mathrm{i}}(\Qbxby)$ are equal.
    For each equivalence class  $[\Qxy]=\{\Qxy,\Qxby,\Qbxy,\Qbxby\}$, we can choose $\Qxy^* \in [\Qxy]$ satisfying\footnote{
        When all the four divergences are infinity, any choice will do.
    }
    \begin{align*}
        D(\Qxy^*) \leq \min[D(\Qxby^*),D(\Qbxy^*),D(\Qbxby^*)].         
    \end{align*}
    Considering that flipped types have same type-class size and the same entropy, $\Qxy^*$ satisfies the following by the probability of type class expression \eqref{eq:probability of type}:
    \begin{equation}\label{eq:inequality of probability of type class }
        P(\mathcal{T}_{\Qxy^*})\leq \left(P(\mathcal{T}_{\Qxy})+P(\mathcal{T}_{\Qxby})+P(\mathcal{T}_{\Qbxy})+P(\mathcal{T}_{\Qbxby})\right)\leq 4P(\mathcal{T}_{\Qxy^*}).
    \end{equation} 
    Using this inequality, $C^{(m)}$ can be bounded from above by an expression involving only $\Qxy^*$ rather than all $\Qxy$ as follows.
    First, note that
    \begin{align}
        C^{(m)} &= \sum_{\Qxy \in \mathcal{P}_m} P(\mathcal{T}_{\Qxy}) C_{\mathrm{i}}(\Qxy)\notag\\ 
        &= \frac{1}{4}\sum_{\Qxy \in \mathcal{P}_m} \left\{P(\mathcal{T}_{\Qxy})C_{\mathrm{i}}(\Qxy)+P(\mathcal{T}_{\Qxby})C_{\mathrm{i}}(\Qxby)+P(\mathcal{T}_{\Qbxy})C_{\mathrm{i}}(\Qbxy)+P(\mathcal{T}_{\Qbxby})C_{\mathrm{i}}(\Qbxby)\right\}\notag\\
        &\overset{(a)}{=} \frac{1}{4}\sum_{\Qxy \in \mathcal{P}_m} \left\{P(\mathcal{T}_{\Qxy})+P(\mathcal{T}_{\Qxby})+P(\mathcal{T}_{\Qbxy})+P(\mathcal{T}_{\Qbxby})\right\} C_{\mathrm{i}}(\Qxy^*),\notag
    \end{align}
    where (a) follows from the fact that $C_{\mathrm{i}}(\Qxy),C_{\mathrm{i}}(\Qxby),C_{\mathrm{i}}(\Qbxy),C_{\mathrm{i}}(\Qbxby)$ are equal.
    Next, we apply \eqref{eq:inequality of probability of type class } to the positive term of $C_{\mathrm{i}}(\Qxy^*)$ to get
    \begin{align*}
        &\left\{P(\mathcal{T}_{\Qxy})+P(\mathcal{T}_{\Qxby})+P(\mathcal{T}_{\Qbxy})+P(\mathcal{T}_{\Qbxby}) \right\} h \left(\frac{2^{-m\Dxyas}+2^{-m\Dbxbyas}}{S(\Qxy^*)}\right)\\
        &\leq 4P(\mathcal{T}_{\Qxy^*}) h\left( \frac{2^{-m\Dxyas}+2^{-m\Dbxbyas}}{S(\Qxy^*)} \right).
    \end{align*}
    On the other hand, we apply \eqref{eq:inequality of probability of type class } to the negative terms of $C_{\mathrm{i}}(\Qxy^*)$ to get
    \begin{align*}
        &\left\{P(\mathcal{T}_{\Qxy})+P(\mathcal{T}_{\Qxby})+P(\mathcal{T}_{\Qbxy})+P(\mathcal{T}_{\Qbxby}) \right\} \left\{\left(\frac{2^{-m\Dxyas}+2^{-m\Dbxbyas}}{S(\Qxy^*)}\right)h\left(\frac{2^{-m\Dxyas}}{2^{-m\Dxyas}+2^{-m\Dbxbyas}}\right)\right.\\
        &\quad\left.+\left(\frac{2^{-m\Dxbyas}+2^{-m\Dbxyas}}{S(\Qxy^*)}\right)h\left(\frac{2^{-m\Dxbyas}}{2^{-m\Dxbyas}+2^{-m\Dbxyas}}\right)\right\}\\
        & \geq P(\mathcal{T}_{\Qxy^*}) \left\{\left(\frac{2^{-m\Dxyas}+2^{-m\Dbxbyas}}{S(\Qxy^*)}\right)h\left(\frac{2^{-m\Dxyas}}{2^{-m\Dxyas}+2^{-m\Dbxbyas}}\right)\right.\\
        &\quad\left.+\left(\frac{2^{-m\Dxbyas}+2^{-m\Dbxyas}}{S(\Qxy^*)}\right)h\left(\frac{2^{-m\Dxbyas}}{2^{-m\Dxbyas}+2^{-m\Dbxyas}}\right)\right\}
    \end{align*}

    Since the upper bounds only involve $\Qxy^*$ rather than all $\Qxy$, we can restrict the range of summation to $\mathcal{Q}_m$ defined in \eqref{eq:mathcalQ_m} to obtain\footnote{For the positive terms, even though some of $\Qxy,\Qxby,\Qbxy$and $\Qbxby$ in a class $[\Qxy]$ might  be the same, it suffices to count four times. In fact, $\mathcal{Q}_m$ may contain more than one type from each class, so this upper involves overcounting. For the negative term, we can just omit those types such that $\Qxy \notin \mathcal{Q}_m$.}
    \begin{align}
        C^{(m)}&\leq \sum_{\Qxy \in \mathcal{P}_m} P(\mathcal{T}_{\Qxy^*}) h\left( \frac{2^{-m\Dxyas}+2^{-m\Dbxbyas}}{S(\Qxy^*)} \right)\notag\\
        &\quad -\frac{1}{4}\sum_{\Qxy \in \mathcal{P}_m} P(\mathcal{T}_{\Qxy^*})\left\{\left(\frac{2^{-m\Dxyas}+2^{-m\Dbxbyas}}{S(\Qxy^*)}\right)h\left(\frac{2^{-m\Dxyas}}{2^{-m\Dxyas}+2^{-m\Dbxbyas}}\right)\right.\notag\\
        &\quad \left.+\left(\frac{2^{-m\Dxbyas}+2^{-m\Dbxyas}}{S(\Qxy^*)}\right)h\left(\frac{2^{-m\Dxbyas}}{2^{-m\Dxbyas}+2^{-m\Dbxyas}}\right)\right\}\notag\\
        & \leq \sum_{\Qxy\in \mathcal{Q}_m} P(\mathcal{T}_{\Qxy}) \left\{ 4h\left( \frac{2^{-m\Dxy}+2^{-m\Dbxby}}{S(\Qxy)} \right) - \frac{1}{4}\left(\frac{2^{-m\Dxy}+2^{-m\Dbxby}}{S(\Qxy)}\right)h\left(\frac{2^{-m\Dxy}}{2^{-m\Dxy}+2^{-m\Dbxby}}\right)\right.\notag\\ 
        &\quad \left.-\frac{1}{4}\left(\frac{2^{-m\Dxby}+2^{-m\Dbxy}}{S(\Qxy)}\right)h\left(\frac{2^{-m\Dxby}}{2^{-m\Dxby}+2^{-m\Dbxy}}\right)\right\}. \label{eq:confined type expression}
    \end{align}

    To get the final upper bound, we need to get an upper bound of the positive term and a lower bound of the negative terms of \eqref{eq:confined type expression}.
    First, we get an upper bound of the positive term.
    To apply the right-hand inequality of binary entropy function
    \begin{align*}
        p \log \frac{1}{p} \leq h(p) \leq 2p \log \frac{1}{p}, \quad p \in [0,1/2],
    \end{align*}
    we distinguish two cases.
    \begin{enumerate}
        \item 
        When $\frac{2^{-m\Dxy}+2^{-m\Dbxby}}{S(\Qxy)} \leq 1/2$, we have
        \begin{align*}
            &h\left(\frac{2^{-m\Dxy}+2^{-m\Dbxby}}{S(\Qxy)}\right)\\
            &\overset{(a)}{\leq} 2\left(\frac{2^{-m\Dxy}+2^{-m\Dbxby}}{S(\Qxy)}\right)\log \frac{S(\Qxy)}{2^{-m\Dxy}+2^{-m\Dbxby}}\\
            &\overset{(b)}{\leq} 2\left(\frac{2^{-m\Dxy}+2^{-m\Dbxby}}{2^{-m\Dxy}}\right) \log \frac{S(\Qxy)}{2^{-m\Dxy}+2^{-m\Dbxby}}\\
            &\overset{(c)}{\leq} 2\left(\frac{2^{-m\Dxy}+2^{-m\Dbxby}}{2^{-m\Dxy}}\right) \left(\log \frac{4}{2^{-m\Dxy}}\right)\\
            &= 2\left(\frac{2^{-m\Dxy}+2^{-m\Dbxby}}{2^{-m\Dxy}}\right) \left(2 +m\Dxy\right)\\
            &\overset{(d)}{\leq} 2\left(\frac{2^{-m\Dxy}+2^{-m\Dbxby}}{2^{-m\Dxy}}\right) \left(2 + m \log \frac{1}{\theta^\star} \right)\\
            &\overset{(e)}{\leq} 2\left(\frac{2^{-m\Dxby}+2^{-m\Dbxy}}{2^{-m\Dxy}}\right) \left(2 + m \log \frac{1}{\theta^\star}\right).
        \end{align*}
        where \\
        (a) follows from the upper bound of binary entropy function, \\
        (b) follows from the fact that $S(\Qxy)\geq 2^{-m\Dxy}$, \\
        (c) follows from the fact that $S(\Qxy) \leq 4$, \\
        (d) follows from the fact that $\Dxy \leq \log \frac{1}{\theta^\star}$, where $\theta^\star$ is defined in \eqref{nt:theta star},\\
        (e) follows from the fact that $2^{-m\Dxy}+2^{-m\Dbxby} \leq 2^{-m\Dxby}+2^{-m\Dbxy}$ in this case.

        Note that $\Dxy$ is finite since $\Qxy$ is a type on $\operatorname{supp}(P_{XY})$.
        Therefore, $\Dxy$ can be upper bounded by $\log \frac{1}{\theta^\star}$.
        \item
            When $\frac{2^{-m\Dxy}+2^{-m\Dbxby}}{S(\Qxy)} > 1/2$, we upper bound
            \begin{align*}
                &h\left(\frac{2^{-m\Dxy}+2^{-m\Dbxby}}{S(\Qxy)}\right) = h\left(\frac{2^{-m\Dxby}+2^{-m\Dbxy}}{S(\Qxy)}\right)\\
                \intertext{There is possibility that both $\Dxby$ and $\Dbxy$ are infinite, however, since $h(0)=0$ when both $\Dxby$ and $\Dbxy$ are infinite, we may assume that either $\Dxby$ or $\Dbxy$ is finite. Under this assumption, we get}
                &h\left(\frac{2^{-m\Dxby}+2^{-m\Dbxy}}{S(\Qxy)}\right)\\
                &\overset{(a)}{\leq} 2\left(\frac{2^{-m\Dxby}+2^{-m\Dbxy}}{S(\Qxy)}\right)\log \frac{S(\Qxy)}{2^{-m\Dxby}+2^{-m\Dbxy}}\\
                &\overset{(b)}{\leq} 2\left(\frac{2^{-m\Dxby}+2^{-m\Dbxy}}{2^{-m\Dxy}}\right) \log \frac{4}{2^{-m\Dxby}+2^{-m\Dbxy}},\\
            \end{align*}
            where \\
            (a) follows from the upper bound of binary entropy function, \\ 
            (b) follows from the fact that $S(\Qxy) \leq 4$ and $S(\Qxy)\geq 2^{-m\Dxy}$.\\
            Let $D^\star$ denote the largest finite divergence among $\Dxby$ and $\Dbxy$.
            Then, we can further upper bound as 
            \begin{align*}            
                &2\left(\frac{2^{-m\Dxby}+2^{-m\Dbxy}}{2^{-m\Dxy}}\right) \log \frac{4}{2^{-m\Dxby}+2^{-m\Dbxy}}\\
                &\overset{(c)}{\leq} 2\left(\frac{2^{-m\Dxby}+2^{-m\Dbxy}}{2^{-m\Dxy}}\right) \left(2 + mD^\star\right)\\
                &\overset{(d)}{\leq} 2\left(\frac{2^{-m\Dxby}+2^{-m\Dbxy}}{2^{-m\Dxy}}\right) \left(2 + m \log \frac{1}{\theta^\star}\right),
            \end{align*}
            where\\
            (c) follows from the fact that $2^{-m\Dxby}+2^{-m\Dbxy}\geq 2^{-mD^\star}$, \\
            (d) follows from the fact that $D^\star \leq \log \frac{1}{\theta^\star}$.
    \end{enumerate}
    Thus, in both cases, we obtain the upper bound as
    \begin{equation}\label{eq:upper bound of positive term}
        h\left( \frac{2^{-m\Dxy}+2^{-m\Dbxby}}{S(\Qxy)} \right) \leq 2\left(\frac{2^{-m\Dxby}+2^{-m\Dbxy}}{2^{-m\Dxy}}\right) \left(2 + m \log \frac{1}{\theta^\star}\right).
    \end{equation}

    We next get the lower bound of negative terms of \eqref{eq:confined type expression} as
    \begin{align}
        &\left(\frac{2^{-m\Dxy}+2^{-m\Dbxby}}{S(\Qxy)}\right)h\left(\frac{2^{-m\Dbxby}}{2^{-m\Dxy}+2^{-m\Dbxby}}\right)+\left(\frac{2^{-m\Dxby}+2^{-m\Dbxy}}{S(\Qxy)}\right)h\left(\frac{2^{-m\Dxby}}{2^{-m\Dxby}+2^{-m\Dbxy}}\right)\notag\\
        & \overset{(a)}{\geq}\left(\frac{2^{-m\Dxy}+2^{-m\Dbxby}}{S(\Qxy)}\right)h\left(\frac{2^{-m\Dbxby}}{2^{-m\Dxy}+2^{-m\Dbxby}}\right) \notag\\
        &\overset{(b)}{\geq} \left(\frac{2^{-m\Dxy}+2^{-m\Dbxby}}{S(\Qxy)}\right) \left(\frac{2^{-m\Dbxby}}{2^{-m\Dxy}+2^{-m\Dbxby}}\right) \log \frac{2^{-m\Dxy}+2^{-m\Dbxby}}{2^{-m\Dbxby}}\notag\\
        &\overset{(c)}{\geq}\left(\frac{2^{-m\Dbxby}}{S(\Qxy)}\right) \log \frac{2\cdot 2^{-m\Dbxby}}{2^{-m\Dbxby}}\notag\\
        &= \left(\frac{2^{-m\Dbxby}}{S(\Qxy)}\right) \cdot 1\notag\\
        &\overset{(d)}{\geq} \left(\frac{2^{-m\Dbxby}}{4\cdot 2^{-m\Dxy}}\right) \notag\\
        &= \frac{2^{-m(\Dbxby-\Dxy)}}{4}, \label{eq:lower bound of negative terms}
    \end{align}
    where \\
    (a) follows from the non-negativity of the second term, \\
    (b) follows from the lower bound of binary entropy function, \\
    (c) and (d) follow from $\Dxy$ is the smallest among $\Dxy,\Dxby,\Dbxy$ and $\Dbxby$ for $\Qxy \in \mathcal{Q}_m$.

    Since the lower bound is 0 when $\Dbxby$ is infinite, the lower bound is also valid for such a case.

    Note that the probability of type class $P(\mathcal{T}_{\Qxy})$ can be bounded as \cite[Lemma 2.6]{Csiszár_Körner_2011}
    \begin{equation}\label{eq:probability of type class}
        \frac{2^{-m\Dxy} }{(m+1)^4}\leq P(\mathcal{T}_{\Qxy}) \leq 2^{-m\Dxy}.
    \end{equation}
    
    Thus, substituting the bounds from \eqref{eq:upper bound of positive term}, \eqref{eq:lower bound of negative terms} and \eqref{eq:probability of type class} into \eqref{eq:confined type expression}, we get the upper bound of $C^{(m)}$ as \eqref{eq:upper bound of C}.
    \qed

\subsection{Preparation for proving lemma \ref{lem:minimum divergence}}
    In this section, we introduce a lemma that will be used in the proof of Lemma \ref{lem:minimum divergence}.
    \begin{lemma}\label{lem:boundary}
    Let $P$ and $P'$ be two probability distributions on the same finite alphabet and suppose that $\operatorname{supp}(P) = \operatorname{supp}(P')$.
    Consider the minimization of $D(Q\|P')$ over all probability distributions $Q$ satisfying
    \begin{align*}
        D(Q\|P)\leq D(Q\|P').
    \end{align*}
    Then the minimum is attained on the boundary
    \begin{align*}
        D(Q\|P)=D(Q\|P').
    \end{align*}
    \end{lemma}

    \begin{proof}
    The case $P=P'$ is trivial, so we assume that $P\neq P'$.
    Suppose, for contradiction, that the minimum is attained at a distribution $Q^*$ satisfying the strict inequality
    \begin{align*}
        D(Q^*\|P)<D(Q^*\|P').
    \end{align*}
    
    For $0 \leq \lambda \leq 1$, define
    \begin{align*}
        Q_\lambda=\lambda Q^*+(1-\lambda)P'.
    \end{align*}
    By the convexity of divergence in the first argument, for $\lambda \in (0,1)$, we have
    \begin{align}
        D(Q_\lambda\|P')&\leq \lambda D(Q^*\|P')+(1-\lambda)D(P'\|P')\notag\\ 
        &=\lambda D(Q^*\|P')\notag\\
        &< D(Q^*\|P'). \label{eq:divergence to P'}
    \end{align}

    Now define
    \begin{align*}
        F(\lambda)=D(Q_\lambda\|P)-D(Q_\lambda\|P').
    \end{align*}
    Since \(D(Q^*\|P')<\infty\), we have
    \[
    \operatorname{supp}(Q^*)\subseteq \operatorname{supp}(P').
    \]
    By the assumption \(\operatorname{supp}(P)=\operatorname{supp}(P')\),
    it follows that
    \[
    \operatorname{supp}(Q_\lambda)\subseteq \operatorname{supp}(P)
    =\operatorname{supp}(P')
    \]
    for every \(\lambda\in[0,1]\). Therefore, both
    \(D(Q_\lambda\|P)\) and \(D(Q_\lambda\|P')\) are finite and continuous
    in \(\lambda\). Hence, \(F(\lambda)\) is continuous on \([0,1]\).
    Since $\operatorname{supp}(P)=\operatorname{supp}(P')$, we have $D(P'\|P)<\infty$.
    Moreover,
    \begin{align*}
        F(1)=D(Q^*\|P)-D(Q^*\|P')<0,
    \end{align*}
    whereas
    \begin{align*}
        F(0)=D(P'\|P)-D(P'\|P')=D(P'\|P)>0.
    \end{align*}
    Hence, by the intermediate value theorem, there exists $\lambda_0\in(0,1)$ such that
    \begin{align*}
        F(\lambda_0)=0.
    \end{align*}
    That is,
    \begin{align*}
        D(Q_{\lambda_0}\|P)=D(Q_{\lambda_0}\|P').
    \end{align*}
    Furthermore,
    \begin{align*}
        D(Q_{\lambda_0}\|P') < D(Q^*\|P')
    \end{align*}
    by \eqref{eq:divergence to P'}.
    Therefore, there exists a feasible distribution on the boundary whose objective value is strictly smaller than that of $Q^*$.
    This contradicts that $Q^*$ is the minimizer. 
    Thus, the minimum is attained on the boundary
    \begin{align*}
        D(Q\|P)=D(Q\|P').
    \end{align*}
    \end{proof}

    \subsection{Proof of Lemma~\ref{lem:minimum divergence}}\label{prf:optimization}
    We first prove the formulas under the full-support assumption $\theta_{ij}>0$ for all $i,j$.
    If exactly one of the source probabilities is zero, the same formulas are obtained by applying Lemma~\ref{lem:boundary} on the corresponding lower-dimensional faces.

    We first prove \eqref{eq:min Dbxby}.
    We first compute the relaxed minimum of $\Dbxby$ under the single constraint $\Dxy\leq\Dbxby$.
    By Lemma~\ref{lem:boundary}, the minimum is attained on the boundary
    \begin{align*}
        \Dxy=\Dbxby.
    \end{align*}
    Let
    \begin{align*}
        \Qxy
        =
        \begin{bmatrix}
            q_{00} & q_{01}\\
            q_{10} & q_{11}
        \end{bmatrix}.
    \end{align*}
    Since
    \begin{align*}
        \Pxy
        =
        \begin{bmatrix}
            \theta_{00} & \theta_{01}\\
            \theta_{10} & \theta_{11}
        \end{bmatrix},
    \end{align*}
    we have
    \begin{align*}
        \Pbxby
        =
        \begin{bmatrix}
            \theta_{11} & \theta_{10}\\
            \theta_{01} & \theta_{00}
        \end{bmatrix}.
    \end{align*}
    Therefore,
    \begin{align*}
        \Dbxby
        =
        q_{00}\log\frac{q_{00}}{\theta_{11}}
        +
        q_{01}\log\frac{q_{01}}{\theta_{10}}
        +
        q_{10}\log\frac{q_{10}}{\theta_{01}}
        +
        q_{11}\log\frac{q_{11}}{\theta_{00}}.
    \end{align*}
    Also, the boundary condition $\Dxy=\Dbxby$ is equivalent to
    \begin{align*}
        \Dxy-\Dbxby=0.
    \end{align*}
    Since
    \begin{align*}
        \Dxy
        =
        q_{00}\log\frac{q_{00}}{\theta_{00}}
        +
        q_{01}\log\frac{q_{01}}{\theta_{01}}
        +
        q_{10}\log\frac{q_{10}}{\theta_{10}}
        +
        q_{11}\log\frac{q_{11}}{\theta_{11}},
    \end{align*}
    we have
    \begin{align*}
        \Dxy-\Dbxby
        =
        q_{00}\log\frac{\theta_{11}}{\theta_{00}}
        +
        q_{01}\log\frac{\theta_{10}}{\theta_{01}}
        +
        q_{10}\log\frac{\theta_{01}}{\theta_{10}}
        +
        q_{11}\log\frac{\theta_{00}}{\theta_{11}}.
    \end{align*}

    For $\lambda,\mu \in \mathbb{R}$, let
    \begin{align*}
        \mathcal{L}
        =
        \Dbxby
        +
        \lambda(q_{00}+q_{01}+q_{10}+q_{11}-1)
        +
        \mu(\Dxy-\Dbxby)
    \end{align*}
    be the Lagrangian for the optimization problem.

    Substituting the above expressions, we obtain
    \begin{align*}
        \mathcal{L}
        &=
        q_{00}\log\frac{q_{00}}{\theta_{11}}
        +
        q_{01}\log\frac{q_{01}}{\theta_{10}}
        +
        q_{10}\log\frac{q_{10}}{\theta_{01}}
        +
        q_{11}\log\frac{q_{11}}{\theta_{00}}
        \\
        &\quad
        +
        \lambda(q_{00}+q_{01}+q_{10}+q_{11}-1)
        \\
        &\quad
        +
        \mu\left(
            q_{00}\log\frac{\theta_{11}}{\theta_{00}}
            +
            q_{01}\log\frac{\theta_{10}}{\theta_{01}}
            +
            q_{10}\log\frac{\theta_{01}}{\theta_{10}}
            +
            q_{11}\log\frac{\theta_{00}}{\theta_{11}}
        \right).
    \end{align*}

    Since the objective function is convex and the constraints are affine,
    any feasible point satisfying the Lagrangian stationary condition is a global minimizer \cite{BV09}.

    Taking the partial derivatives with respect to $q_{00},q_{01},q_{10},q_{11}$, we get
    \begin{align*}
        \frac{\partial\mathcal{L}}{\partial q_{00}}
        &=
        \log\frac{q_{00}}{\theta_{11}}
        +
        \lambda'
        +
        \mu\log\frac{\theta_{11}}{\theta_{00}}
        =0,
        \\
        \frac{\partial\mathcal{L}}{\partial q_{01}}
        &=
        \log\frac{q_{01}}{\theta_{10}}
        +
        \lambda'
        +
        \mu\log\frac{\theta_{10}}{\theta_{01}}
        =0,
        \\
        \frac{\partial\mathcal{L}}{\partial q_{10}}
        &=
        \log\frac{q_{10}}{\theta_{01}}
        +
        \lambda'
        +
        \mu\log\frac{\theta_{01}}{\theta_{10}}
        =0,
        \\
        \frac{\partial\mathcal{L}}{\partial q_{11}}
        &=
        \log\frac{q_{11}}{\theta_{00}}
        +
        \lambda'
        +
        \mu\log\frac{\theta_{00}}{\theta_{11}}
        =0,
    \end{align*}
    where $\lambda'$ absorbs $\lambda$ and the additive constant coming from the derivative of $q\log q$.

    Solving these equations, we obtain
    \begin{align*}
        q_{00}&=K\theta_{00}^{\mu}\theta_{11}^{1-\mu},\quad
        q_{01}=K\theta_{01}^{\mu}\theta_{10}^{1-\mu},\\
        q_{10}&=K\theta_{10}^{\mu}\theta_{01}^{1-\mu},\quad
        q_{11}=K\theta_{11}^{\mu}\theta_{00}^{1-\mu},
    \end{align*}
    where $K$ is the normalizing constant.

    The boundary condition $\Dxy=\Dbxby$ is satisfied by $\mu=1/2$.
    Hence, the minimizing distribution $\Qxy^*$ is
    \begin{align*}
        q_{00}^*=q_{11}^*
        =
        \frac{\sqrt{\theta_{00}\theta_{11}}}{B_{XY}},\quad
        q_{01}^*=q_{10}^*
        =
        \frac{\sqrt{\theta_{01}\theta_{10}}}{B_{XY}},
    \end{align*}
    where
    \begin{align*}
        B_{XY}
        =
        2\left(
            \sqrt{\theta_{00}\theta_{11}}
            +
            \sqrt{\theta_{01}\theta_{10}}
        \right).
    \end{align*}

    We next check that this distribution is also in $\mathcal{Q}$.

    By construction,
    \begin{align*}
        \Dxyas=\Dbxbyas.
    \end{align*}
    Moreover,
    \begin{align*}
        \Dxbyas-\Dxyas
        &=
        \frac{\sqrt{\theta_{00}\theta_{11}}}{B_{XY}}
        \log\frac{\theta_{00}\theta_{11}}{\theta_{01}\theta_{10}}
        +
        \frac{\sqrt{\theta_{01}\theta_{10}}}{B_{XY}}
        \log\frac{\theta_{01}\theta_{10}}{\theta_{00}\theta_{11}}
        \\
        &=
        \frac{
            \sqrt{\theta_{00}\theta_{11}}
            -
            \sqrt{\theta_{01}\theta_{10}}
        }{
            B_{XY}
        }
        \log
        \frac{\theta_{00}\theta_{11}}{\theta_{01}\theta_{10}}
        \geq 0,
    \end{align*}
    where we used $(x-y)\log(x/y)\geq 0$.  Similarly,
    \begin{align*}
        \Dbxyas-\Dxyas
        &=
        \frac{\sqrt{\theta_{00}\theta_{11}}}{B_{XY}}
        \log\frac{\theta_{00}\theta_{11}}{\theta_{01}\theta_{10}}
        +
        \frac{\sqrt{\theta_{01}\theta_{10}}}{B_{XY}}
        \log\frac{\theta_{01}\theta_{10}}{\theta_{00}\theta_{11}}
        \\
        &=
        \frac{
            \sqrt{\theta_{00}\theta_{11}}
            -
            \sqrt{\theta_{01}\theta_{10}}
        }{
            B_{XY}
        }
        \log
        \frac{\theta_{00}\theta_{11}}{\theta_{01}\theta_{10}}
        \geq 0.
    \end{align*}
    Therefore,
    \begin{align*}
        \Dxyas
        \leq
        \min\{\Dxbyas,\Dbxyas,\Dbxbyas\}.
    \end{align*}
    Thus,
    \begin{align*}
        \Qxy^\ast\in\mathcal{Q}.
    \end{align*}
    Furthermore, we can verify that $Q_{XY}^*$ satisfies \eqref{eq:boundary equality} and \eqref{eq:boundary equality2} from the above argument.
    Consequently, the exponent
    \begin{align*}
        \log\frac{1}{B_{XY}}
    \end{align*}
    is achievable by a distribution in $\mathcal{Q}$.

    Substituting this distribution into $\Dbxby$, we obtain
    \begin{align*}
        \Dbxbymin
        =
        \min_{\Qxy\in\mathcal{Q}} \Dbxby
        =
        \log\frac{1}{B_{XY}}
        =
        \log
        \frac{1}{
            2\left(
                \sqrt{\theta_{00}\theta_{11}}
                +
                \sqrt{\theta_{01}\theta_{10}}
            \right)
        }.
    \end{align*}

    In the same manner, we can compute the minimum of $\Dxby$ and $\Dbxy$.
    Consequently, we have \eqref{eq:min Dxby} and \eqref{eq:min Dbxy}.

    It remains to consider the case where exactly one of the source probabilities is zero. 
    We describe only one representative case, since the other cases are obtained by flipping $X$, $Y$, or both.
    Without loss of generality, assume that
    \begin{align*}
        \theta_{00}=0,
        \qquad
        \theta_{01}=p,
        \qquad
        \theta_{10}=q,
        \qquad
        \theta_{11}=r,
    \end{align*}
    where $p,q,r>0$ and $p+q+r=1$.
    For any type $\Qxy$ that can occur under $\Pxy$, we have
    \begin{align*}
        q_{00}=0.
    \end{align*}
    Since $\Dxy$ is smallest among $\Dxy,\Dxby,\Dbxy$and$\Dbxby$,
    it is enough to consider distributions $\Qxy$ satisfying $\Dxy<\infty$.

    First, consider the minimization of $\Dbxby$.
    Since
    \begin{align*}
        \Pxy
        =
        \begin{bmatrix}
            0 & p\\
            q & r
        \end{bmatrix},
        \qquad
        \Pbxby
        =
        \begin{bmatrix}
            r & q\\
            p & 0
        \end{bmatrix},
    \end{align*}
    the finiteness of $\Dbxby$ further requires
    \begin{align*}
        q_{11}=0.
    \end{align*}
    Hence, the optimization is restricted as
    \begin{align*}
        q_{00}=q_{11}=0,
        \qquad
        q_{01}+q_{10}=1.
    \end{align*}
    On this constraint, the two relevant supports coincide and are equal to $\{(0,1),(1,0)\}$.
    Let $\widetilde{Q}_{XY}$ be the binary distribution induced by $\Qxy$, i.e.,
    \begin{align*}
        \widetilde{Q}_{XY}
        =
        (q_{01},q_{10}).
    \end{align*}
    Also define the normalized restrictions of $\Pxy$ and $\Pbxby$ to this constraint by
    \begin{align*}
        \widetilde{P}_{XY}
        =
        \left(
            \frac{p}{p+q},
            \frac{q}{p+q}
        \right),
        \qquad
        \widetilde{P}'_{XY}
        =
        \left(
            \frac{q}{p+q},
            \frac{p}{p+q}
        \right).
    \end{align*}
    Then, for any distribution $\Qxy$ on this constraint, we have
    \begin{align*}
        \Dxy
        &=
        q_{01}\log\frac{q_{01}}{p}
        +
        q_{10}\log\frac{q_{10}}{q}
        \\
        &=
        D(\widetilde{Q}_{XY}\|\widetilde{P}_{XY})
        -
        \log(p+q),
        \\
        \Dbxby
        &=
        q_{01}\log\frac{q_{01}}{q}
        +
        q_{10}\log\frac{q_{10}}{p}
        \\
        &=
        D(\widetilde{Q}_{XY}\|\widetilde{P}'_{XY})
        -
        \log(p+q).
    \end{align*}
    Thus, Lemma~\ref{lem:boundary} can be applied between $D(\widetilde{Q}_{XY}\|\widetilde{P}'_{XY})$ and $D(\widetilde{Q}_{XY}\|\widetilde{P}_{XY})$.
    The minimizer $\Qxy^*$ is attained on the boundary and is given by
    \begin{align*}
        q_{01}^*=q_{10}^*=\frac{1}{2}.
    \end{align*}
    Substituting this distribution into $\Dbxbyas$, we obtain
    \begin{align*}
        \Dbxbyas
        &=
        \frac{1}{2}\log\frac{1/2}{q}
        +
        \frac{1}{2}\log\frac{1/2}{p}
        \\
        &=
        \log\frac{1}{2\sqrt{pq}}.
    \end{align*}

    We next check that the distribution $\Qxy^*$ is also in $\mathcal{Q}$.
    By construction, $\Dxy=\Dbxby$.

    On the other hand, $\Dxby=\infty$ because $q_{01}>0$ while the
    corresponding probability in $\Pxby$ is zero.  Similarly,
    $\Dbxy=\infty$ because $q_{10}>0$ while the corresponding
    probability in $\Pbxy$ is zero.  Hence,
    \[
        \Dxy
        \leq
        \min\{\Dxby,\Dbxy,\Dbxby\},
    \]
    and therefore this minimizer belongs to $\mathcal{Q}$.
    We can also verify that $Q_{XY}^*$ satisfies \eqref{eq:boundary equality} and \eqref{eq:boundary equality2}.

    The same lower-dimensional argument applies to the minimization of $\Dxby$ and $\Dbxy$.
    Thus, the same formulas as in the full-support case remain valid when exactly one source probability is zero. 
    The cases $\theta_{01}=0$, $\theta_{10}=0$, and $\theta_{11}=0$ follow in the same way after flipping $X$, $Y$, or both.

\subsection{Density of types}\label{prf:density of types}
    \begin{lemma}
        There exists a type $Q_{m,XY}\in\mathcal{Q}_m$ and $\epsilon_m$ such that
        $\epsilon_m\to 0$ as $m\to\infty$
        and 
        \begin{equation}\label{eq:density of types}
            D(Q_{m,XY}\|\Pbxby)
            \leq
            \Dbxbymin + \epsilon_m
        \end{equation}
        for sufficiently large $m$.
    \end{lemma}
    \begin{proof}
    First, we show that  there exists a sequence of joint types $Q_m$ such that $D(Q_m\|\Pbxby)\to D(Q^\ast\|\Pbxby)$ as $m\to\infty$ where
    $Q^\ast\in\mathcal{Q}$ is a minimizer of $D(Q\|\Pbxby)$ over $Q\in\mathcal{Q}$.
    Note that $\operatorname{supp}(Q^*) \subseteq \operatorname{supp}(\Pbxby)$,
    since $D(Q^*\|\Pbxby)$ is infinite otherwise.
    Let $\|\cdot\|$ denote the total variation distance defined as
    \begin{align*}
        \|Q-Q'\|
        =
        \sum_{x,y}
        |Q(x,y)-Q'(x,y)|.
    \end{align*}

    First, note that there exists a sequence of joint types $Q_m$ of length $m$ such that $\operatorname{supp}(Q_m)=\operatorname{supp}(Q^\ast)$ and
    \begin{align*}
        \|Q_m-Q^\ast\|
        \leq
        \frac{|\mathcal{X}\times\mathcal{Y}|}{m}=\frac{4}{m}.
    \end{align*}

    For $P' \in \{\Pxy,\Pxby,\Pbxy,\Pbxby\}$, we show that $D(Q_m\|P')\to D(Q^\ast\|P')$ as $m\to\infty$ as long as $D(Q^\ast\|P')<\infty$.
    The difference can be bounded as
    \begin{align*}
        \left|D(Q_m\|P')-D(Q^\ast\|P')\right|
        &\leq
        \left|H(Q_m)-H(Q^\ast)\right| 
        +
        \left|
            \sum_{x,y}
            \left(Q_m(x,y)-Q^\ast(x,y)\right)
            \log P'(x,y)
        \right|.
    \end{align*}
    By the continuity of entropy \cite[Lemma 2.7]{Csiszár_Körner_2011},
    \begin{align*}
        \left|H(Q_m)-H(Q^\ast)\right|
        \leq
        \|Q_m-Q^\ast\|
        \log
        \frac{|\mathcal{X}\times\mathcal{Y}|}{\|Q_m-Q^\ast\|}.
    \end{align*}
    Moreover, since every probability of $P'$ is at least $\theta^\star$ on $\operatorname{supp}(P')$,
    \begin{align*}
        \left|
            \sum_{x,y}
            \left(Q_m(x,y)-Q^\ast(x,y)\right)
            \log P'(x,y)
        \right|
        \leq
        \|Q_m-Q^\ast\| \log\frac{1}{\theta^\star}.
    \end{align*}
    Hence,
    \begin{align*}
        \left|D(Q_m\|P')-D(Q^\ast\|P')\right|
        \leq
        \|Q_m-Q^\ast\|
        \log
        \frac{|\mathcal{X}\times\mathcal{Y}|}{\|Q_m-Q^\ast\|}
        +
        \|Q_m-Q^\ast\| \log\frac{1}{\theta^\star}.
    \end{align*}
    Since $\|Q_m-Q^\ast\|\leq 4/m$, and since
    \[
        t\log\frac{4}{t}
    \]
    is increasing for sufficiently small $t$, we obtain, for sufficiently large $m$,
    \begin{align*}
        \|Q_m-Q^\ast\|
        \log
        \frac{4}{\|Q_m-Q^\ast\|}
        \leq
        \frac{4}{m}\log m.
    \end{align*}
    Therefore,
    \begin{equation}\label{eq:continuity of divergence}
        \left|D(Q_m\|P')-D(Q^\ast\|P')\right|
        \leq
        \frac{4}{m}\log m
        +
        \frac{4}{m}\log\frac{1}{\theta^\star}
        \eqcolon \epsilon_m.
    \end{equation}
    Since the above bound is uniform over
    \(P'\in\{\Pxy,\Pxby,\Pbxy,\Pbxby\}\)
    whenever the corresponding divergence is finite, the same
    sequence of types \(Q_m\) approximates all relevant divergences
    simultaneously.
    Thus, for any fixed minimizer \(Q^\ast\), we can choose a
    sequence of joint types \(Q_m\) such that all the relevant
    finite divergences converge to the corresponding divergences
    of \(Q^\ast\).

    Next, we need to check that the approximating type $Q_m$ can be chosen
    from $\mathcal{Q}_m$ for sufficiently large $m$.


    Let $Q_m$ be a type sufficiently close to $Q^\ast$, i.e., $\|Q_m - Q^\ast\| \leq \frac{4}{m}$.
    Then, we choose $Q_m^\pi$ which satisfies \[D(Q_m^\pi\|\Pxy) \leq \min[D(Q_m^\pi\|\Pxby), D(Q_m^\pi\|\Pbxy), D(Q_m^\pi\|\Pbxby)].\]
    Here, $\pi$ is defined as any of the four flipping operations: no flip, flipping $X$, flipping $Y$, or flipping both.
    We can show that $D(Q_m^\pi\|\Pbxby)$ is also sufficiently close to $D(Q^\ast\|\Pbxby)$ as follows:
    \begin{enumerate}
        \item When $D(Q^\ast \| \Pbxby) < D(Q^\ast \| \Pxby)$,
        \begin{align*}
            D(Q_m^\pi\| \Pbxby) &= D(Q_m\| \Pbxby^\pi) \\
            &\overset{(a)}{\leq} D(Q^\ast \| \Pbxby^\pi) + \epsilon_m \\
            &\overset{(b)}{=} D(Q^\ast \| \Pbxby) + \epsilon_m,
        \end{align*}
        where \\
        (a) follows from the continuity of divergence \eqref{eq:continuity of divergence},\\
        (b) follows from the fact that (see \eqref{eq:boundary equality} in Lemma \ref{lem:minimum divergence})\[D(Q^\ast||\Pxy) = D(Q^\ast||\Pbxby)\]
        and $Q_m^\pi$ only takes non-flipped version or both-flipped version of $Q_m$ for sufficiently large $m$
        because $D(Q^\ast \| \Pbxby) < D(Q^\ast \| \Pxby)$.
        
        \item When $D(Q^\ast \| \Pbxby) = D(Q^\ast \| \Pxby)$,
        \begin{align*}
            D(Q_m^\pi\| \Pbxby) &= D(Q_m\| \Pbxby^\pi) \\
            &\overset{(c)}{\leq} D(Q^\ast \| \Pbxby^\pi) + \epsilon_m \\
            &\overset{(d)}{= } D(Q^\ast \| \Pbxby) + \epsilon_m,
        \end{align*}
        where \\
        (c) follows from the continuity of divergence \eqref{eq:continuity of divergence},\\
        (d) follows from the fact that \[D(Q^\ast||\Pxy) =D(Q^\ast||\Pxby)=D(Q^\ast||\Pbxy)=D(Q^\ast||\Pbxby).\]
        See \eqref{eq:boundary equality} and \eqref{eq:boundary equality2} in Lemma \ref{lem:minimum divergence}.
    \end{enumerate}
    Thus, we can choose $Q_m^\pi\in Q_m$ satisfying \eqref{eq:density of types}.
    \end{proof}

%% file: sections/acknowledgment.tex
\section*{Acknowledgment}
The authors would like to thank Prof. Hirosuke Yamamoto for bringing \cite{yamamoto:82,BergerHousewright:79,Housewright:77} and \cite{YamamotoItoh:1980} to their attention 
and for teaching them about the historical background of Berger-Tung coding.

This work was supported in part by the Japan Society for the Promotion of Science (JSPS) KAKENHI under Grant Numbers
23K17455, and 26H02489, and 26K02864, and by JST, CRONOS, Japan Grant Number JPMJCS25N5.

ChatGPT was used to assist with the minimization calculation in Lemma 8. It was also used for language editing and improving the presentation of the manuscript.